\documentclass[envcountsect,envcountsame]{llncs}
\usepackage{amssymb,stmaryrd}
\usepackage{amsmath}
\usepackage{epsfig}
\usepackage{makeidx}
\usepackage{enumitem}
\usepackage{array}
\usepackage{algorithmic}
\usepackage[ruled]{algorithm}

\algsetup{indent=2em}
\renewcommand{\algorithmicendif}{\textbf{fi}}

\newcommand{\OLIF}[2]
{\STATE \algorithmicif \mbox{ #1 }\algorithmicthen\mbox{ #2 }\algorithmicendif}

\begin{document}

\title{\mbox{Generalized Maneuvers in Route Planning}}
\author{Petr Hlin\v{e}n\'y ~\and~ Ondrej Mori\v{s}}

\institute{Faculty of Informatics, Masaryk University \\
  Botanick\'a 68a, 602 00 Brno, Czech Republic  \\
  \email{hlineny@fi.muni.cz, xmoris@fi.muni.cz}}

\maketitle

\begin{abstract}
We study an important practical aspect of the route planning problem in 
real-world road networks -- \emph{maneuvers}. Informally, maneuvers 
represent various irregularities of the road network graph such as
turn-prohibitions, traffic light delays, round-abouts, forbidden passages 
and so on. We propose a generalized model which can handle arbitrarily complex
(and even negative) maneuvers, and outline how to enhance Dijkstra's 
algorithm in order to solve route planning queries in this model without 
prior adjustments of the underlying road network graph.
\end{abstract}

\section{Introduction}

Since mass introduction of GPS navigation devices, the \emph{route planning 
problem}, has received considerable attention. This problem is in fact
an instance of the well-known single pair shortest path (SPSP) problem
in graphs representing real-world road networks. However, it involves 
many challenging difficulties compared to ordinary SPSP. Firstly, classical 
algorithms such as Dijkstra's \cite{Dijkstra1959}, A* \cite{Hart1972} or their 
bidirectional variants \cite{Pohl1969} are not well suited for the route
planning despite their optimality in wide theoretical sense. It is mainly 
because graphs representing real-world road networks are so huge that even 
an algorithm with linear time and space complexity cannot be feasibly run 
on typical mobile devices. 

Secondly, these classical approaches disregard certain important aspects of
real-world road networks, namely route restrictions, traffic regulations, or
actual traffic info. Hence a route found by such algorithms might not be 
optimal or not even feasible. Additional attributes are needed in this regard.

The first difficulty has been intensively studied in the past decade, and 
complexity overheads of classical algorithms have been largely improved by 
using various preprocessing approaches. For a brief overview, we refer the 
readers to \cite{Cherkassky1994,Delling2009B,Schultes2008} or our 
\cite{HM2011A}. In this paper we focus on the second mentioned difficulty 
as it is still receiving significantly less attention.

\subsubsection{Related Work. } 
The common way to model required additional attributes of road networks is
with so called {\em maneuvers}\/; Definition~\ref{def:maneuver}. Maneuvers 
do not seem to be in the center of interest of route-planning research papers:
They are either assumed to be encoded into the underlying graph of a road 
network, or they are addressed only partially with rather simple types of 
restriction attributes such as turn-penalties and path prohibitions.

Basically, the research directions are represented either by modifications 
of the underlying graph during preprocessing 
\cite{Jiang2002,Pallottino1997,Ziliaskopoulos1996}, or by adjusting a query 
algorithm \cite{Kirby1969,Villeneuve2005} in order to resolve simple types 
of restrictions during queries. 

The first, and seemingly the simplest, solution is commonly used as it makes 
a~road network graph maneuver-free and so there is no need to adjust the 
queries in any way. Unfortunately, it can significantly increase the size 
of the graph \cite{Winter2002}; for instance, replacing a single 
turn-prohibition can add up to eight new vertices in place of one 
original~\cite{Gutierrez2008}. A solution like this one thus conflicts with 
the aforementioned (graph-size) objectives. Another approach \cite{Anez1996} 
uses so-called dual graph representation instead of the original one, where 
allowed turns are modeled by dual edges. 

To summarize, a sufficiently general approach for arbitrarily complex 
maneuvers seems to be missing in the literature despite the fact that such 
a solution could be really important. We would like to emphasize that all 
the cited works suffer from the fact that they consider only ``simple'' types 
of maneuvers.

\vspace{-1.5ex}
\subsubsection{Our Contribution. } Firstly, we introduce a formal model of a 
generic maneuver -- from a single vertex to a long self-intersecting walk 
-- with either positive or negative effects (penalties); being enforced, 
recommended, not recommended or even prohibited. Our model can capture 
virtually any route restriction, most traffic regulations and even some 
dynamic properties of real-world road networks. 

Secondly, we integrate this model into Dijkstra's algorithm, rising its 
worst-case time complexity only slightly (depending on a structure of 
maneuvers). The underlying graph is not modified at all and no preprocessing 
is needed. Even though our idea is fairly simple and relative easy to 
understand, it is novel in the respect that no comparable solution has been 
published to date. Furthermore, some important added benefits of our 
algorithm are as follows:

\begin{itemize}
\parskip 2pt
\item It can be directly used bidirectionally with any alternation strategy
  using an appropriate termination condition; it can be extended also to the 
  A* algorithm by applying a ``potential function to maneuver effects''.
  
\item Many route planning approaches use Dijkstra or A* in the core of their 
  query algorithms, and hence our solution can be incorporated into many of 
  them (for example, those based on a reach, landmarks or various types of 
  separators) quite naturally under additional assumptions.

\item Our algorithm tackles maneuvers ``on-line'' -- that is\ no maneuver 
  is processed before it is reached. And since the underlying graph of a 
  road network is not changed (no vertices or edges are removed or added),
  it is possible to add or remove maneuvers dynamically even during queries
  to some extent.
\end{itemize}

\section{Maneuvers: Basic Terms}
\label{sec:maneuvers}

A \emph{(directed) graph} $G=(V,E)$ is a pair of a finite set $V$ of vertices 
and a finite multiset $E \subseteq V \times V$ of edges (self-loops and 
parallel edges are allowed). The vertex set of $G$ is referred to as $V(G)$, 
its edge multiset as $E(G)$. \emph{A~subgraph} $H$ of a graph $G$ is denoted 
by $H \subseteq G$. 

\emph{A walk} $P \in G$ is an alternating sequence of vertices and 
edges $(u_0,e_1,u_1,\ldots,$ $e_k,u_k) \subseteq G$ such that $e_i = 
(u_{i-1},u_i)$ for $i = 1,\ldots,k$, the multiset of all edges of a walk $P$
is denoted by $E(P)$. \emph{A concatenation} $P_1 .\,P_2$ of 
walks $P_1=(u_0,e_1,u_1,\ldots,e_k,u_k)$ and $P_2=(u_k,e_{k+1},u_{k+1},\ldots,
e_l,u_l)$ is the walk $(u_0,e_1,u_1,$ $\ldots,e_k,u_k,e_{k+1},\ldots,e_l,u_l)$. 
If $P_2=(u,f,v)$ represents a single edge, we write~$P_1.\,f$. If edges are 
clear from the graph, then we write a walk simply as $(u_0,u_1,\dots,u_k)$.

A walk $Q$ is a {\em prefix} of another walk $P$ if $Q$ is a subwalk of
$P$ starting with the same index, and analogically with {\em suffix}.
The \emph{prefix set} of a walk $P=(u_0,e_1,\ldots,e_k,u_k)$ is $\mathit{Prefix}
(P) = \{(u_0,e_1,\ldots,e_i,u_i) |\> 0 \le i \le k\},$ and analogically 
$\mathit{Suf\!fix}(P)=\{(u_i,e_{i+1},\ldots,e_k,u_k) |\>  0 \le i\le k\}$. 
A prefix (suffix) of a walk $P$ thus is a member of $\mathit{Prefix}(P)$
($\mathit{Suf\!fix}(P)$), and it is {\em nontrivial} if $i\geq1$.

\emph{The weight} of a walk $P \subseteq G$ with respect to a weighting $w: 
E(G) \mapsto \mathbb{R}$ of $G$ is defined as $\sum_{e \in E(P)} w(e)$ and 
denoted by $|P|_w$. \emph{A distance} from $u$ to $v$ in $G$, $\delta_w(u,v)$, 
is the minimum weight of a walk $P=(u, \ldots, v) \subseteq G$ over all such 
walks and $P$ is then called \emph{optimal} (with respect to weighting $w$). 
If there no such walk then $\delta_w(u,v) = \infty$. \emph{A~path} is a walk 
without repeating vertices and edges.

Virtually any route restriction or traffic regulation in a road network, 
such as turn-prohibitions, traffic lights delays, forbidden passages, 
turn-out lanes, suggested directions or car accidents by contrast, can be 
modeled by \emph{maneuvers} -- walks having extra (either positive or negative) 
``cost effects''. Formally:

\begin{definition}[Maneuver]
\label{def:maneuver}
\emph{A maneuver} $M$ of $G$ is a walk in $G$ that is assigned a penalty 
$\Delta(M)\! \in\! \mathbb{R}\!\cup\! \infty$. A set of all maneuvers of $G$ is 
denoted by ${\cal{M}}$.
\end{definition}

\begin{remark}
A maneuver with a negative or positive penalty is called \emph{negative} 
or \emph{positive}, respectively. Furthermore, there are two special kinds 
of maneuvers the \emph{restricted} ones of penalty 0 and the \emph{prohibited} 
ones of penalty~$\infty$. 
\end{remark}

The cost effect of a maneuver is formalized next:
\begin{definition}[Penalized Weight]
\label{def:penalized_weight}
Let $G$ be a graph with a weighting $w$ and a set of maneuvers $\cal{M}$. 
The \emph{penalized weight} of a walk $P \subseteq G$ containing the maneuvers 
$M_1,\ldots,M_r \in {\cal{M}}$ as subwalks is defined as $|P|_w^{\cal{M}} = 
|P|_w + \sum_{i=1}^{r} \Delta(M_i)$.
\end{definition}

Then, the intended meaning of maneuvers in route planning is as follows.

\begin{itemize}
\parskip2pt
\item If a driver enters a restricted maneuver, she must pass it completely
  (cf.~Definition~\ref{def:valid_walk}); 
  she must obey the given direction(s) regardless of the cost effect.
  Examples are headings to be followed or specific round-abouts.

\item By contrast, if a driver enters a prohibited maneuver, she must 
  not pass it completely. She must get off it before reaching its end, 
  otherwise it makes her route infinitely bad. Examples are forbidden 
  passages or temporal closures.

\item Finally, if a driver enters a positive or negative maneuver, she is
  not required to pass it completely;
  but if she does, then this will increase or decrease the 
  cost of her route accordingly. Negative maneuvers make her route 
  better (more desirable) and positive ones make it worse.
  Examples of positive maneuvers are, for instance, traffic lights delays, lane
  changes, or left-turns.
  Examples of negative ones are turn-out lanes, shortcuts, or implicit routes.
\end{itemize}

\begin{definition}[Valid Walk]
\label{def:valid_walk}
Let $G,w,\cal{M}$ be as in Definition~\ref{def:penalized_weight}. A walk 
$P$ in $G$ is \emph{valid} if and only if $|P|_w^{\cal{M}} < \infty$ and, for any
restricted maneuver $M\in\cal{M}$, it holds that if a nontrivial prefix of
$M$ is a subwalk of $P$, then whole $M$ is a~subwalk of $P$ or a suffix of $P$
is contained in~$M$ (that is $P$ ends there).
\end{definition}

We finally get to the summarizing definition. A structure of a road network 
is naturally represented by a graph $G$ such that the junctions are 
represented by $V(G)$ and the roads by $E(G)$. The chosen cost function 
(for example travel time, distance, expenses) is represented by a \emph{non-negative}
weighting $w: E(G) \mapsto\mathbb{R}^+_0$ assigned to $G$, and the additional 
attributes such as traffic regulations are represented by maneuvers as above.
We say that two walks $Q_1,Q_2$ are {\em divergent} if, up to symmetry
between $Q_1,Q_2$, a nontrivial prefix of $Q_1$ is contained in $Q_2$ but
the whole $Q_1$ is not a subwalk of~$Q_2$. Moreover, we say that $Q_2$ 
{\em overhangs} $Q_1$ if a~nontrivial prefix of $Q_2$ is a suffix of~$Q_1$
(particularly, $E(Q_1)\cap E(Q_2)\neq \emptyset$).

\begin{definition}[Road Network]
\label{def:road_network}
Let $G$ be a graph with a non-negative weighting $w$ and a set of maneuvers 
$\cal M$. \emph{A road network} is the triple $(G,w,{\cal M})$. 
Furthermore, it is called \emph{proper} if:
\begin{itemize}
\parskip 2pt
\item[i.] no two restricted maneuvers in $\cal M$ are divergent,
\item[ii.] no two negative maneuvers in $\cal M$ overhang one another, and
\item[iii.] for all $N \in {\cal{M}}$, $\Delta(N) \geq -|N|_w^{{\cal{M}} 
    \setminus \{N\}}$ (that is, the penalized weight of every walk in $G$ is 
  non-negative). 
\end{itemize}
\vspace{-4pt}
\end{definition}

\begin{figure}[t]
  \centering
  \centerline{\epsfig{file=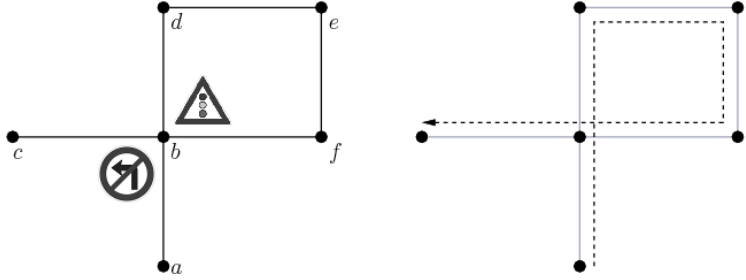, scale=0.6}}
  \vspace*{-.5ex}
  \caption{A road network containing maneuvers $M_1=(a,ab,b,bc,c)$ 
    with $\Delta(M_1)=\infty$ (prohibited left turn) and $M_2=(a,ab,b,bf,f)$ 
    with $\Delta(M_2) = 1$ (right turn traffic lights delay). All edges
    have weight 1. The penalized weight of the walk $(a,ab,b,bc,c)$ is $2 + 
    \infty$, the penalized weight of the walk $(a,ab,b,bf,f,fe,e,ed,d,db,b,bd,
    c)$ is $6 + 1$. Therefore the optimal walk (with respect to the 
    penalized weight) from $a$ to $c$ is $(a,ab,b,bd,d,de,e,ef,f,fb,b,bc,c)$ 
    with the penalized weight $6 + 0$.}
  \label{fig:maneuver}
\end{figure}

Within a road network, only valid walks (Definition~\ref{def:valid_walk})
are allowed further, and the distance from $u$ to $v$, $\delta_w^{\cal{M}}
(u,v)$, is the minimum penalized weight (Definition~\ref{def:penalized_weight})
of a valid walk $P=(u, \ldots, v) \subseteq G$; such a walk $P$ is then 
called \emph{optimal with respect to the penalized weight}. If there is no 
such walk, then $\delta_w^{\cal{M}}(u,v) = \infty$. See~Fig.~\ref{fig:maneuver}.

Motivation for the required properties i.--iii. in 
Definition~\ref{def:road_network} is of both natural and practical character:
As for i., it simply says that no two restricted maneuvers are in a
conflict (that is no route planning deadlocks). Point ii. concerning only 
negative maneuvers is needed for a fast query algorithm, and it is indeed 
a~natural requirement (to certain extent, overhanging maneuvers can be modeled 
without overhangs). We remark that other studies usually allow no negative 
maneuvers at all. Finally, iii. states that no negative maneuvers can result 
in a negative overall cost of any walk -- another very natural property.
In informal words, a negative penalty of a maneuver somehow ``cannot 
influence'' suitability of a~route before entering and after exiting the 
maneuver.

\subsection{Strongly Connected Road Network}
\label{sec:connectivity}

The traditional graph theoretical notion of strong connectivity also needs to
be refined, it must suit our road networks to dismiss possible route 
planning traps now imposed by maneuvers. 

First, we need to define a notion of a \emph{``context''} of a vertex $v$ 
in $G$ -- a~maximal walk in $G$ ending at $v$ such that it is a proper 
prefix of a maneuver in $\cal M$, or $\emptyset$ otherwise. A set of all such
walks for $v$ is denoted by ${\cal{X}}_{\cal{M}}$. For example, on the road 
network depicted on Fig.~\ref{fig:maneuver}, 
${\cal{X}}_{\cal M}(b)=\{(a,b),\emptyset\}$. More formally: 

\begin{definition}
  \label{def:context}
  Let $\cal M$ be a set of maneuvers. We define 
  $$ {\cal{X}}_{{\cal{M}}}(v) \stackrel{\textrm{\tiny{def}}}{=} 
  \big\{ X \in \mathit{Prefix}^{<}({\cal M}) \, | (v) \in 
  \mathit{Suf\!fix}(X)\big\} \cup \{\emptyset\} $$
  $$ \mathit{Prefix}^{<}(M) \stackrel{\textrm{\tiny{def}}}{=} \mathit{Prefix}(M)
  \setminus \{M\}, \quad \mathit{Prefix}^{<}({\cal M}) 
  \stackrel{\textrm{\tiny{def}}}{=} \bigcup\nolimits_{M\in 
    \cal M}\mathit{Prefix}^{<}(M). $$
  This ${\cal{X}}_{{\cal{M}}}(v)$ is the {\em maneuver-prefix set} at $v$,
  that is\ the set of all proper prefixes of walks from $\cal M$ that end
  right at~$v$, including the mandatory empty walk.
  An element of ${\cal{X}}_{{\cal{M}}}(v)$ is called a {\em context} of the
  position $v$ within the road network.
\end{definition}

\emph{The reverse graph} $G^R$ of $G$ is a graph on the same set of vertices 
with all of the edges reversed. Let $(G,w,{\cal{M}})$ be a road network, 
a \emph{reverse road network} is defined as $(G^R,w^R,{\cal{M}}^R)$, where 
$w^R: E(G^R) \mapsto \mathbb{R}_0^+$, $\forall (u,v) \in E(G^R):\, w^R(u,v)
= w(v,u)$ and ${\cal{M}}^R = \{ M^R | M \in {\cal{M}}\}$, $\forall M^R \in 
{\cal{M}}^R:\, \Delta(M^R) = \Delta(M)$.

\begin{definition}
  \label{def:Mconnectivity}
  A road network $(G,w,{\cal M})$ is \emph{strongly connected} if, for every 
  pair of edges $e=(u',u),\>f=(v,v') \in E(G)$ and for each possible 
  context $X=X_1\cdot\,e\in{\cal X}_{\cal M}(u)$ of $u$ in $G$ and each one 
  of $v$ in $G^R$, that is \ $Y^R=Y_1^R.\,f^R \in 
  {\cal X}_{{\cal M}^R}(v)$, there exists a valid walk starting with $X$ and 
  ending with~$Y$.
\end{definition}

We remark that Definition~\ref{def:Mconnectivity} naturally corresponds to 
strong connectivity in an~amplified road network modeling the maneuvers 
within underlying graph.

\section{Route Planning Queries}
\label{sec:query}

At first, let us recall classical Dijkstra's algorithm \cite{Dijkstra1959}.
It solves SPSP\footnote{Given a graph and two vertices find a shortest path 
from one to another.} problem a graph $G$ with a~non-negative weighting $w$ 
for a pair $s,t\in V(G)$ of vertices.

\begin{itemize}
\parskip 3pt
\item The algorithm maintains, for all $v \in V(G)$, a 
  {\em (temporary) distance estimate} of the shortest path from $s$ to $v$ 
  found so far in $d[v]$, and a predecessor of $v$ on that path in $\pi[v]$. 

\item The scanned vertices, that is those with $d[v] = \delta_w(s,v)$, 
  are stored in the set $T$; and the reached but not yet scanned vertices, 
  that is those with $\infty >d[v] \geq \delta_w(s,v)$, are stored in the set $Q$. 

\item The algorithm work as follows: it iteratively picks a vertex $u \in Q$ 
  with minimum value $d[u]$ and relaxes all the edges $(u,v)$ leaving $u$.
  Then $u$ is removed from $Q$ and added to $T$. {\em Relaxing} an edge $(u,v)$
  means to check if a shortest path estimate from $s$ to $v$ may be improved 
  via $u$; if so, then $d[v]$ and $\pi[v]$ are updated. Finally, $v$ is added 
  into $Q$ if is not there already. 

\item The algorithm terminates when $t$ is scanned or when $Q$ is empty.
\end{itemize}

Time complexity depends on the implementation of $Q$; such as it is 
${\cal{O}}(|E(G)| + |V(G)|\log|V(G)|)$ with the Fibonacci heap.

\subsection{$\cal{M}$-Dijkstra's Algorithm}
\label{sec:Mdijkstra}

In this section we will briefly sketch the core ideas of our natural extension 
of Dijkstra's algorithm. We refer a reader to Algorithm \ref{alg:m-dijkstra}
for a full-scale pseudocode of this $\cal{M}$-Dijkstra's algorithm.

\begin{enumerate}
\parskip 3pt
\item Every vertex $v\in V(G)$ scanned during the algorithm is considered 
  together with its context $X \in {\cal X}_{\cal M}(v)$ 
  (Definition~\ref{def:context}); that is\ as a pair $(v,X)$.
  The intention is for $X$ to record how $v$ has  been reached in the 
  algorithm, and same $v$ can obviously be reached and scanned more than 
  once, with different contexts.
  For instance, $b$ can be reached with the empty or $(a,b)$ contexts 
  on the road network depicted on Fig.~\ref{fig:maneuver}.

\item Temporary distance estimates are stored in the algorithm as $d[v,X]$
  for such vertex-context pairs $(v,X)$. At each step the 
  algorithm selects a next pair $(u,Y)$ such that it is minimal with respect 
  to the following partial order $\le_{\cal M}$.

  \begin{remark}{Partial order $\le_{\cal M}$:}
  \label{rmk:order}
  $$
  (v_1,X_1) \le_{\cal M} (v_2,X_2) 
  \stackrel{\textrm{\tiny{def}}}{\Longleftrightarrow}
  \left.\begin{array}{ll}\big(d[v_1,X_1] < d[v_2,X_2] &\lor\\[2pt] 
      (d[v_1,X_1] = d[v_2,X_2] &\land \> X_1 \in \mathit{Suf\!fix}(X_2)\,)\big).
    \end{array}\right.
  $$
  \end{remark}

\item Edge relaxation from a selected vertex-context pair $(u,Y)$ respects
  all maneuvers related to the context $Y$ (there can be more such maneuvers).
  If one of them is restricted, then only its unique (cf.\ Definition
  \ref{def:road_network},\,i.) subsequent edge is taken, cf.\ 
  Algorithm~\ref{alg:m-dijkstra}, \textsc{RestrictedDirection}.

  Otherwise, every edge $f=(u,v)$ is relaxed such that the distance estimate 
  at $v$ -- together with its context as derived from the concatenation 
  $(Y.\,f)$ -- is (possibly) updated with the weight $w(f)$ plus the sum of 
  penalties of all the maneuvers in $(Y.\,f)$ ending at $v$, 
  cf.~Algorithm~\ref{alg:m-dijkstra}, \textsc{Relax}.

\item If an edge relaxed is the first one of a negative maneuver,
  a specific process is executed before scanning the next vertex-context pair.
  See below.
\end{enumerate}

\subsection{Processing Negative Maneuvers} 
\label{sec:process_negative_maneuvers}

Note that the presence of a maneuver of negative penalty
{\em may violate} the basic assumption of ordinary Dijkstra's algorithm;
that relaxing an edge never decreases the nearest temporary distance
estimate in the graph. An example of such a violation can be seen in
Fig.~\ref{fig:negative-maneuver}, for instance, at vertex $v_5$ which would 
not be processed in its correct place by ordinary Dijkstra's algorithm.
That is why a negative maneuver $M$ must be processed by $\cal{M}$-Dijkstra's 
algorithm at once -- whenever its starting edge is relaxed, 
cf.~Algorithm~\ref{alg:m-dijkstra}, \textsc{ProcessNegative}.

Suppose that an edge $f=(u,v)$ is relaxed from a selected vertex-context pair 
$(u,X)$ and there is a negative maneuver $M = (v_0,f_1,v_1,\ldots,v_{n-1},f_n,
v_n)$, $u=v_0$, $v=v_1$ starting with $f$ (that is $f=f_1$), processing negative
maneuver $M$ works as follow:

\begin{figure}[t]
  \centering
  \centerline{\epsfig{file=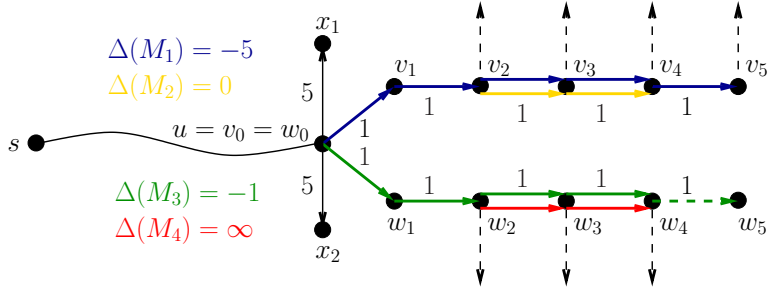, scale=0.6}}
  \caption{A road network containing two negative maneuvers,
    $M_1=(v_0,\ldots,v_5)$ and $M_3=(w_0,\ldots,,w_5)$,
    a restricted maneuver $M_2=(v_2,v_3,v_4)$, and a prohibited 
    maneuver $M_4=(w_2,w_3,w_4)$.
    When $u$ is being processed (with its implicit context),
    $x_1,x_2$ and $v_1,w_1$ are relaxed normally. Furthermore,
    negative maneuver processing is executed for both $M_1$ and $M_3$.
    As a result, $v_5$ will be immediately reached and inserted to $Q$ with
    distance estimate equal to that of $u$ which is less than those of $x_1,
    x_2$ (5 from $u$) and of $v_1,w_1$ (1 from $u$). On the other hand,
    $w_5$ will not be reached in the process because the distance estimate 
    of $w_4$ bounces to $\infty$ while handling $M_4$.}
  \label{fig:negative-maneuver}
\end{figure}

\begin{enumerate}
\parskip 2pt
\item Vertex-context pairs $(v_i,X_i), 0 \le i \le n$ along $M$ are 
  scanned one by one towards the end of $M$. The other vertices leaving 
  these $v_i$ are ignored.

\item Scanned vertex-context pairs are added to $Q$ and their distance 
  estimates are updated, but none of them is added into $T$. They must be 
  properly scanned during the main loop of the algorithm.

\item This process terminates when the end $(v_n,X_n)$ is reached or the 
  distance estimate of some $(v_i,X_i)$ bounces to $\infty$ (that is there is a 
  prohibited maneuver ending at $v_i$) or when some restricted maneuver 
  forces us to get off $M$ (and thus $M$ cannot be completed).
\end{enumerate}

\begin{algorithm}[H]
\caption{~$\cal{M}$-Dijkstra's Algorithm}
\label{alg:m-dijkstra}
\begin{algorithmic}[1]
\smallskip
\REQUIRE A proper road network $(G,w,{\cal M})$ and vertices $s,t \in V(G)$.
\ENSURE A valid walk from $s$ to $t$ in $G$ optimal with respect to the 
penalized weight.
\smallskip
\end{algorithmic}
\textsc{${\cal{M}}$-Dijkstra}$(G,w,{\cal{M}},s,t)$
\begin{algorithmic}[1]
  \small
  \FORALL[\hfill /* Initialization. */]{$v \in V(G)$, $X \in 
    {\cal X}_{\cal M}(v)$}
    \STATE $d[v,X] \leftarrow \infty;~ \pi[v,X] \leftarrow \bot$
  \ENDFOR
  \STATE $d[s,\emptyset] \leftarrow 0$;~ $Q \leftarrow \{(s,\emptyset)\}$;~ $T 
  \leftarrow \emptyset$
  \OLIF{$(s) \in {\cal{M}}$}{
    $d[s,\emptyset] \leftarrow d[s,\emptyset] + \Delta(s)$}
  \medskip 
  \vskip 0pt
  \COMMENT{/* The main loop starts at $(s,\emptyset)$ and terminates
    when either all reachable vertex-\\ \hfill context pairs have been scanned
    or when $t$ is reached with some of its contexts. */}
  \smallskip
  \WHILE{$Q \neq \emptyset \land [\not\exists\, X \in {\cal{X}}_{\cal{M}}(t)$ 
    s.t. $(t,X) \in T]$}
    \smallskip
    \STATE $(u,X) \leftarrow  \min_{\le_{\cal M}}(Q)$;~$Q \leftarrow Q \setminus
    \{(u,X)\}$
    \COMMENT{\hfill /* Recall $\le_{\cal{M}}$ (Remark \ref{rmk:order}) */}
    \smallskip
    \STATE $F \leftarrow \textsc{RestrictedDirection}(u,X)$
    \COMMENT{\hfill /* Possible restricted dir.\ from $u$. */}
    \OLIF{$F = \emptyset$}{
      $F \leftarrow \{ (u,v) \in E(G) \, | \, v \in V(G)\}$}
    \smallskip
    \FORALL{$ f=(u,v) \in F$}
      \STATE \textsc{Relax}$(u,X,f,v)$
      \FORALL{$M = (u,f,v,\ldots) \in {\cal{M}}$ s.t.
		 $\Delta(M)<0 \,\wedge\, |E(M)| > 1$}
        \STATE \textsc{ProcessNegative}$\,(X,M)$ 
      \ENDFOR
    \COMMENT{\hfill /* Negative man.\ starting with $f$ are processed 
      separately. */}
    \ENDFOR
    \STATE $T \leftarrow T \cup \{(u,X)\}$ 
  \ENDWHILE
  \STATE \textsc{ConstructWalk}$\,(G,d,\pi)$
    \COMMENT{\hfill /* Use ``access'' information stored in $\pi[v,X]$. */}
    \smallskip
\end{algorithmic}
\end{algorithm}

\vspace{-4ex}
\small{
\noindent
\textsc{LongestPrefix}$\,(P) :\> \textrm{a walk } P' \subseteq G$
\begin{algorithmic}[1]
    \item[]
    \COMMENT{\hfill /* The longest (proper) prefix of some maneuver contained
      as a suffix of $P$ */}\vskip 0pt
    \smallskip
    \STATE $\begin{array}{l} P' \leftarrow \max_{\subseteq}       
      \big[\,(\mathit{Suf\!fix}(P) \cap \mathit{Prefix}^<({\cal{M}})) 
      \cup \{\emptyset\}\big] \\ ~~~~~~~~~~~~~~~~~~~
      \mbox{where }\mathit{Prefix}^{<}({\cal M}) 
      \stackrel{\textrm{\tiny{def}}}{=} \bigcup\nolimits_{M\in 
        \cal M}\mathit{Prefix}(M)\setminus\{M\} \end{array}$
    \RETURN $P'$
    \smallskip
\end{algorithmic}

\noindent
\textsc{RestrictedDirection}$(u,X) :\> F\subseteq E(G)$
\begin{algorithmic}[1]    
    \item[]
    \COMMENT{\hfill /* Looking for edge $f$ leaving $u$ that follows
			in a restricted man.\ in context $X$.*/}\vskip 0pt
    \STATE $\begin{array}{l}F \leftarrow \{ f = (u,v) \in E(G) \, | \, 
        ~\exists \textrm{ restricted } R \in {\cal{M}}: \\ ~~~~~~~~~~~~~~~~~~~
        E(X)\cap E(R)\neq \emptyset \, \land 
        \mathit{Suf\!fix}(X.\, f) \cap \mathit{Prefix}(R)\neq \emptyset\}
        \end{array}$
    \RETURN $F$
\bigskip
\end{algorithmic}

\noindent
\textsc{Relax}$\,(u,X,f,v)$
\begin{algorithmic}[1]
    \vspace*{-\baselineskip}    
    \item[]
    \COMMENT{\hfill /* Relaxing an edge $f$ from vertex $u$ with context $X$. 
      */}\vskip 0pt
    \STATE $\delta \leftarrow w(f) + \sum_{N \in {\cal{N}}} \Delta(N)$ ~where 
    ${\cal{N}} = {\cal{M}} \cap \mathit{Suf\!fix}(X .\,f)$      
    \STATE $X' \leftarrow \textsc{LongestPrefix}(X . f)$
    \IF{$d[u,X] + \delta < d[v,X']$}
    \STATE $Q \leftarrow Q \cup \{(v,X')\}$;~
    $d[v,X'] \leftarrow d[u,X] + \delta$;~
    $\pi[v,X'] \leftarrow (u,X)$
    \ENDIF
    \smallskip
\end{algorithmic}

\noindent
\textsc{ProcessNegative}$(X,\, M = (v_0,e_1,\ldots,e_n,v_n)\,)$
\begin{algorithmic}[1]    
  \STATE $i \leftarrow 1; \, X_0 \leftarrow X; \, F \leftarrow \emptyset$ 
  \COMMENT{\hfill /* Relaxing sequentially all the edges of $M$. */}
  \WHILE{$i \leq n \land d[v_{i-1},X_i] <\infty \land F = \emptyset$}
      \STATE \textsc{Relax}$(v_{i-1},X_{i-1},e_i,v_i)$
      \STATE $X_{i}\leftarrow \textsc{LongestPrefix}(X_{i-1} . e_i)$;
      $F \leftarrow \textsc{RestrictedDirection}(v_i,X_i)
		\setminus\{e_{i+1}\}$
      \STATE $i \leftarrow i + 1$
  \ENDWHILE
\end{algorithmic}
}

\subsection{Correctness and Complexity Analysis} 
\label{sec:analysis}

Assuming validity of Definition \ref{def:road_network}\, ii. in a~proper road 
network, correctness of above $\cal{M}$-Dijkstra's algorithm can be argued 
analogously to a traditional proof of Dijkstra's algorithm. Hereafter, the 
time complexity growth of the algorithm depends solely on the number of 
vertex-context pairs.

\begin{theorem}
  \label{thm:MDijkstra}
  Let a proper road network $(G,w,\cal{M})$ and vertices $s,t \in
  V(G)$ be given. ${\cal{M}}$-Dijk\-stra's algorithm (Algorithm~\ref{alg:m-dijkstra})
  computes a valid walk from $s$ to $t$ in $G$ optimal with respect to the
  penalized weight, in time ${\cal O} \big(c^2_{\cal M}|E(G)| + c_{\cal M}|V(G)|
  \log(c_{\cal M}|V(G)|) \big)$ where $c_{\cal M} =\max_{v\in V(G)}| 
  \{M\in{\cal{M}}\, |\, v\in V(M)\}|$ is the maximum number of maneuvers per 
  vertex.
\end{theorem}

\begin{proof}
We follow a traditional proof of ordinary Dijkstra's algorithm with a~simple 
modification -- instead of vertices we consider vertex-context pairs as in 
Definition~\ref{def:Mconnectivity} and in Algorithm~\ref{alg:m-dijkstra}.

For a walk $P$ let $\chi(P)=\max_{\subseteq} \big[\,(\mathit{Suf\!fix}(P) \cap 
\mathit{Prefix}^<({\cal{M}})) \cup \{\emptyset\}\big]$ denote the context of 
the endvertex of $P$ with respect to maneuvers $\cal M$. Let $P_x$ stand for 
the prefix of $P$ up to a vertex $x\in V(P)$. The following invariant holds at 
every iteration of the algorithm:

\begin{itemize}
\item[i.] For every $(u,X)\in T$, the final distance estimate $d[u,X]$ equals 
  the smallest penalized weight of a valid walk $P$ from $s$ to $u$ such that
  $X=\chi(P)$. Every vertex-context pair directly accessible from a member
  of~$T$ belongs to $Q$.
  \medskip
\item[ii.] For every $(v,X')\in Q$, the temporary distance estimate 
  $d[v,X']$ equals the smallest penalized weight of a walk $R$ from 
  $s$ to $v$ such that $X'=\chi(R)$ and, moreover,
  $(x,\chi(R_x))\in T$ for each internal vertex $x\in V(R)$
  (except vertices reached during \textsc{ProcessNegative}, if any).
\end{itemize}

This invariant is trivially true after the initialization. By induction we 
assume it is true at the beginning of the while loop on line~6, and line 7 is 
now being executed -- selecting the pair $(u,X)\in Q$. Then, by minimality of 
this selection, $(u,X)$ is such that the distance estimate $d[u,X]$ gives the 
optimal penalized weight of a walk $P$ from $s$ to $u$ such that 
$X=\chi(P)$. Hence the first part of the invariant (concerning $T$, line~16)
will be true also after finishing this iteration.

Concerning the second claim of the invariant, we have to examine the effect
of lines~8--15 of the algorithm. Consider an edge $f=(u,v) \in E(G)$ starting 
in~$u$, and any walk $R$ from $s$ to $v$ such that $\chi(R_u)=X$. Since 
$\chi(R)$ must be contained in $X.\,f$ by definition; it is, \textsc{Relax}, 
line 2, $\chi(R)=X'$. Furthermore, every maneuver contained in $R$ and not 
in $R_u$ must be a suffix of $X.\,f$ by definition. So the penalized weight 
increase $\delta$ is correctly computed in \textsc{Relax}, line 1. Therefore, 
\textsc{Relax} correctly updates the temporary distance estimate $d[v,X']$ 
for every such~$f$. Finally, any negative maneuver starting from $u$ along 
$f$ is correctly reached towards its end $w$ on line 13, its distance estimate
is updated by successive relaxation of its edges and, by 
Definition~\ref{def:road_network}, ii. and iii., this distance estimate of $w$ and 
its context is not smaller than $d[u,X]$; thus the second part of claimed 
invariant remains true. 

Validity of a walk is given by line 8 -- \textsc{RestrictedDirection}, that is 
enforcing entered restricted maneuvers; and line 1 in \textsc{Relax} --
$\delta$ grows to infinity when completing prohibited maneuvers, 
``if'' condition on line 3 in \textsc{Relax} is then false and therefore
prohibited maneuver cannot be contained in an optimal walk.

\medskip

Lastly, we examine the worst-case time complexity of this algorithm. We 
assume $G$ is efficiently implemented using neighborhood lists, the maneuvers 
in $\cal M$ are directly indexed from all their vertices and their number is 
polynomial in the graph size, and that $Q$ is implemented as Fibonacci heap.

\begin{itemize}
\item The maximal number of vertex-context pairs that may enter $Q$ is
  $$ m= |V(G)|+\sum_{M\in\cal M}(|M| - 1)\leq c_{\cal{M}} \cdot |V(G)| \,,$$
  and time complexity of the Fibonacci heap operations is $O(m\log m)$.
  \medskip

\item Every edge of $G$ starting in $u$ is relaxed at most those many times as 
  there are contexts in ${\cal X_M}(u)$ and edges of negative maneuvers 
  are relaxed one more time during \textsc{ProcessNegative}. Hence the 
  maximal overall number of  relaxations is
  $$ r= \sum_{u\in V(G)}|{\cal X_M}(u)|\cdot\mathit{out\mbox-deg}(u) + q
  \leq (c_{\cal{M}} + 1) \cdot|E(G)|$$
  where $q$ is the number of edges belonging to negative maneuvers.
  \medskip

\item The operations in \textsc{Relax} on line 1, \textsc{LongestPrefix} 
  as well as \textsc{RestrictedDirection} can be implemented in time 
  $O(c_{\cal{M}})$.
\end{itemize}
The claimed runtime bound follows.
\vskip 0pt
\qed
\medskip 
\end{proof}

Notice that, in real-world road networks, the number $c_{\cal M}$ of maneuvers 
per vertex is usually quite small and independent of the road network size, 
and thus it can be bounded by a reasonable minor constant. Although road 
networks in practice may have huge maneuver sets, particular maneuvers do 
not cross or interlap too much there. for example, $c_{\cal M}=5$ in the current 
OpenStreetMaps of Prague.

\subsection{Route Planning Example}
\label{sec:example}

In this section we will demonstrate $\cal{M}$-Dijkstra's algorithm on 
a road network containing maneuvers. Consider the road network depicted 
below with a weighting representing travel times. There are five 
maneuvers (their edges are depicted by dotted lines): a~traffic jam detour 
($M_1$), a forbidden passage ($M_2$), a~traffic light left turn delay 
($M_3$), a~traffic light delay ($M_4$) and a~direct to be followed ($M_5$). 

\medskip
\begin{tabular}{m{170pt}m{150pt}}
\epsfig{file=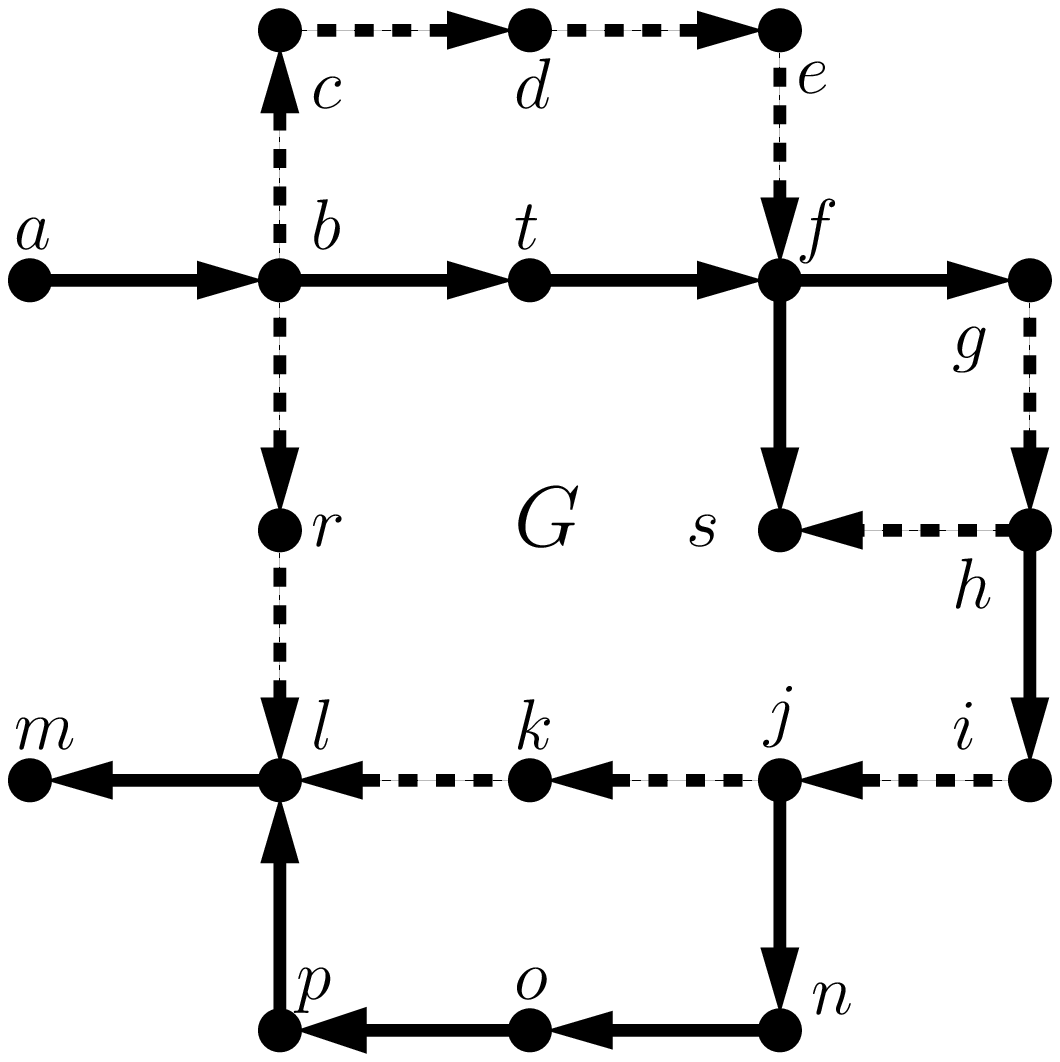, scale=0.44}
&
\hskip 5pt Road network $(G,w,{\cal{M}})$, where
\begin{itemize}
\item ~$G$ is depicted on the left,
\vskip 2pt
\item ~$\forall e \in E(G): w(e)=1$,
\vskip 2pt
\item ~${\cal{M}}=\{M_1,M_2,M_3,M_4,M_5\}$
\vskip 5pt
\begin{tabular}{l@{~}l}
$M_1 = (b,c,d,e,f)$ & $\Delta(M_1)=-3$ \\[3pt]
$M_2 = (b,r,l)$ & $\Delta(M_2)=~\infty$ \\[3pt]
$M_3 = (g,h,s)$ & $\Delta(M_3)=~~5$ \\[3pt]
$M_4 = (s)$ & $\Delta(M_4) = ~~9$ \\[3pt]
$M_5 = (i,j,k,l)$ & $\Delta(M_5)=~~0$
\end{tabular}
\end{itemize} 
\end{tabular}
\medskip

The goal of out driver is to get from $a$ to $m$ as fast as possible. Classical 
Dijkstra's algorithm finds \mbox{$P_1=(a,b,r,l,m)$} with $|P_1|_w = 4$, 
unfortunately $|P_1|_w^{\cal{M}} = \infty$ and hence it is impossible for our 
driver -- it contains a forbidden passage ($M_2$). On the other hand, 
$\cal{M}$-Dijkstra's algorithm finds $P_2=(a,b,c,d,e,f,g,h,i,j,k,l,m)$ with
$|P_2|_w^{\cal{M}} = 9$ and $P_2$ is optimal w.r.t. the penalized weight. 
Steps are outlined in Tab.~\ref{tab:example} and Fig.~\ref{fig:example}. 

\vspace{-2ex}
\begin{table}[H]
\label{tab:example}
\caption{State of selected data structures during the steps of 
Alg.~\ref{alg:m-dijkstra}. Second column shows a vertex-context pair chosen
at the beginning of the while-loop, i.e. $ \min_{\le_{\cal{M}}}(Q)$. Third column
shows its final distance estimate, i.e. $d[u,X]=\delta_w^{\cal{M}}(a,u)$ and,
finally, the last column depicts elements of the queue $Q$ at the end of 
the while-loop.}
\centering
\begin{tabular}{|@{~~}c@{~~}|@{~~}c@{~~}|@{~~}c@{~~}|@{~}c|}
\hline
\textit{Step} & $(u,X)$ \textit{(line 7)} & $d[u,X]$ & $Q$ \textit{(line 16)} \\[2pt]
\hline\hline
1 & $[a,\emptyset]$ & 0 & $[b,\emptyset]$ \\[2pt]
2 & $[b,\emptyset]$ & 1 & $[c,(b,c)]; [t,\emptyset],[r,(b,r)]$ \\[2pt]
3 & $[c,(b,c)]$ & 2 & $[t,\emptyset]; [r,(b,r)]; [d,(b,c,d)]; [e,(b,c,d,e)]; 
[f,\emptyset]$ \\[2pt]
4 & $[t,\emptyset]$ & 2 & $[(r,(b,r)]; [d,(b,c,d)]; [e,(b,c,d,e)]; [f,\emptyset]$ 
\\[2pt]
5 & $[r,(b,r)]$ & 2 & $[d,(b,c,d)]; [e,(b,c,d,e)]; [f,\emptyset]; [l,\emptyset]$ 
\\[2pt]
6 & $[f,\emptyset]$ & 2 & $[g,\emptyset]; [s,\emptyset]; [d,(b,c,d)]; 
[e,(b,c,d,e)]; [l,\emptyset]$ \\[2pt]
7 & $[d,(b,c,d)]$ & 3 & $[g,\emptyset]; [s,\emptyset]; [e,(b,c,d,e)]; 
[l,\emptyset]$ \\[2pt]
8 & $[g,\emptyset]$ & 3 & $[h,(g,h)]; [s,\emptyset]; [e,(b,c,d,e)]; [l,\emptyset]$ 
\\[2pt]
9 & $[e,(b,c,d,e)]$ & 4 & $[h,(g,h)]; [s,\emptyset]; [l,\emptyset]$ \\[2pt]
10 & $[h,(g,h)]$ & 4 & $[i,\emptyset]; [s,\emptyset]; [l,\emptyset]$ \\[2pt]
11 & $[i,\emptyset]$ & 5 & $[j,(i,j)]; [s,\emptyset]; [l,\emptyset]$ \\[2pt]
12 & $[j,(i,j)]$ & 6 & $[k,(i,j,k)]; [s,\emptyset]; [l,\emptyset]$ \\[2pt]
13 & $[k,(i,j,k)]$ & 7 & $[s,\emptyset]; [l,\emptyset]$ \\[2pt]
14 & $[l,\emptyset]$ & 8 & $[m,\emptyset]; [s,\emptyset]$ \\[2pt]
15 & $[m,\emptyset]$ & 9 & $[s,\emptyset]$ \\
\hline
\end{tabular}
\end{table}

\begin{figure}[H]
  \centering
  \epsfig{file=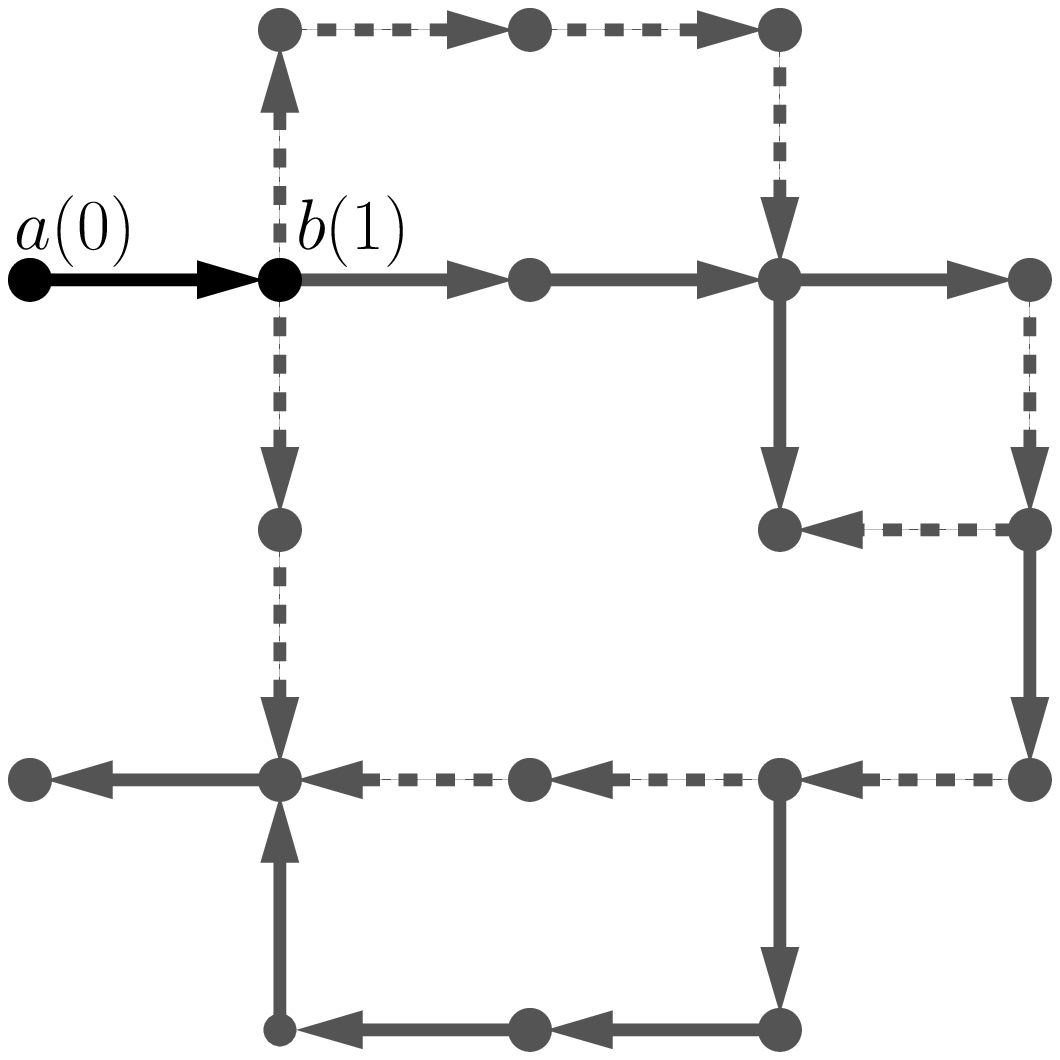, scale=0.33}
  ~
  \epsfig{file=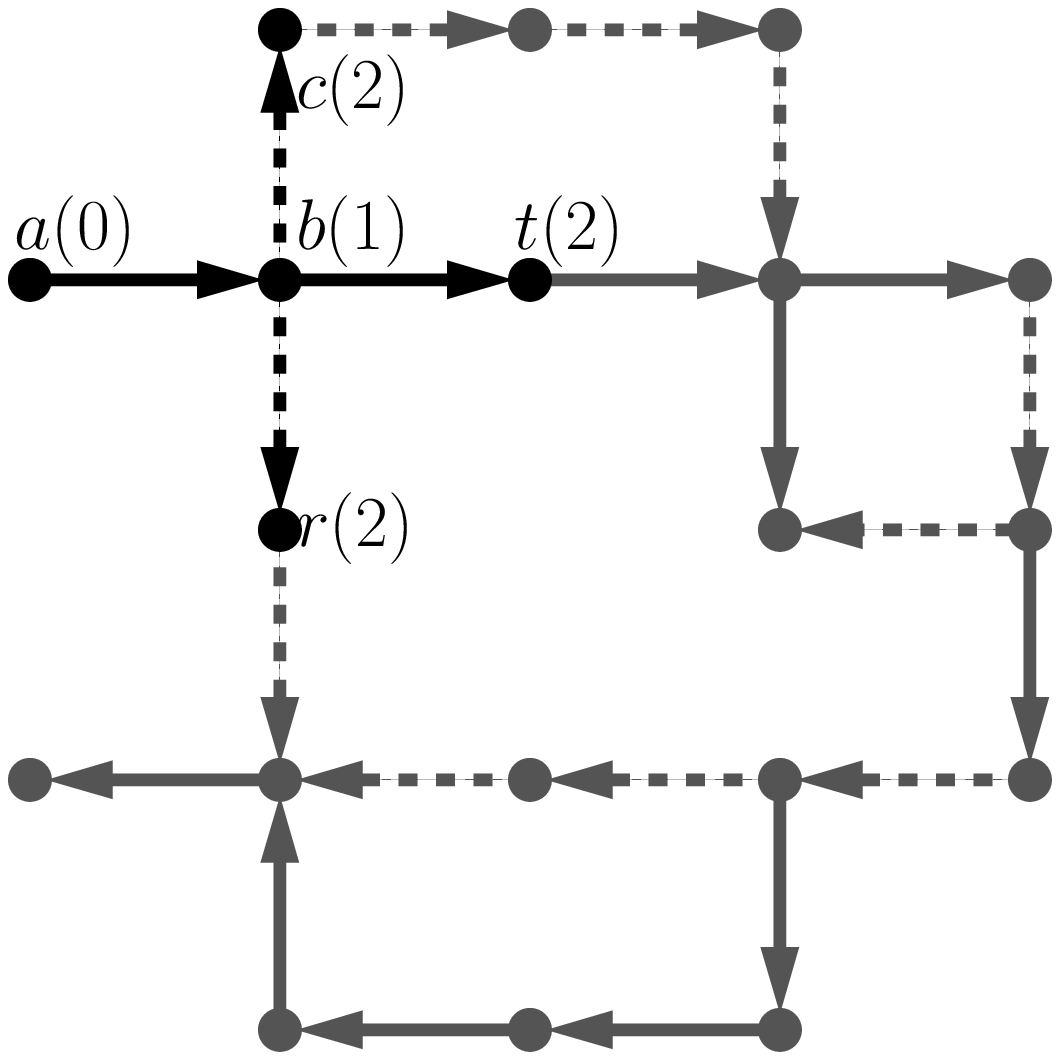, scale=0.33}
  ~
  \epsfig{file=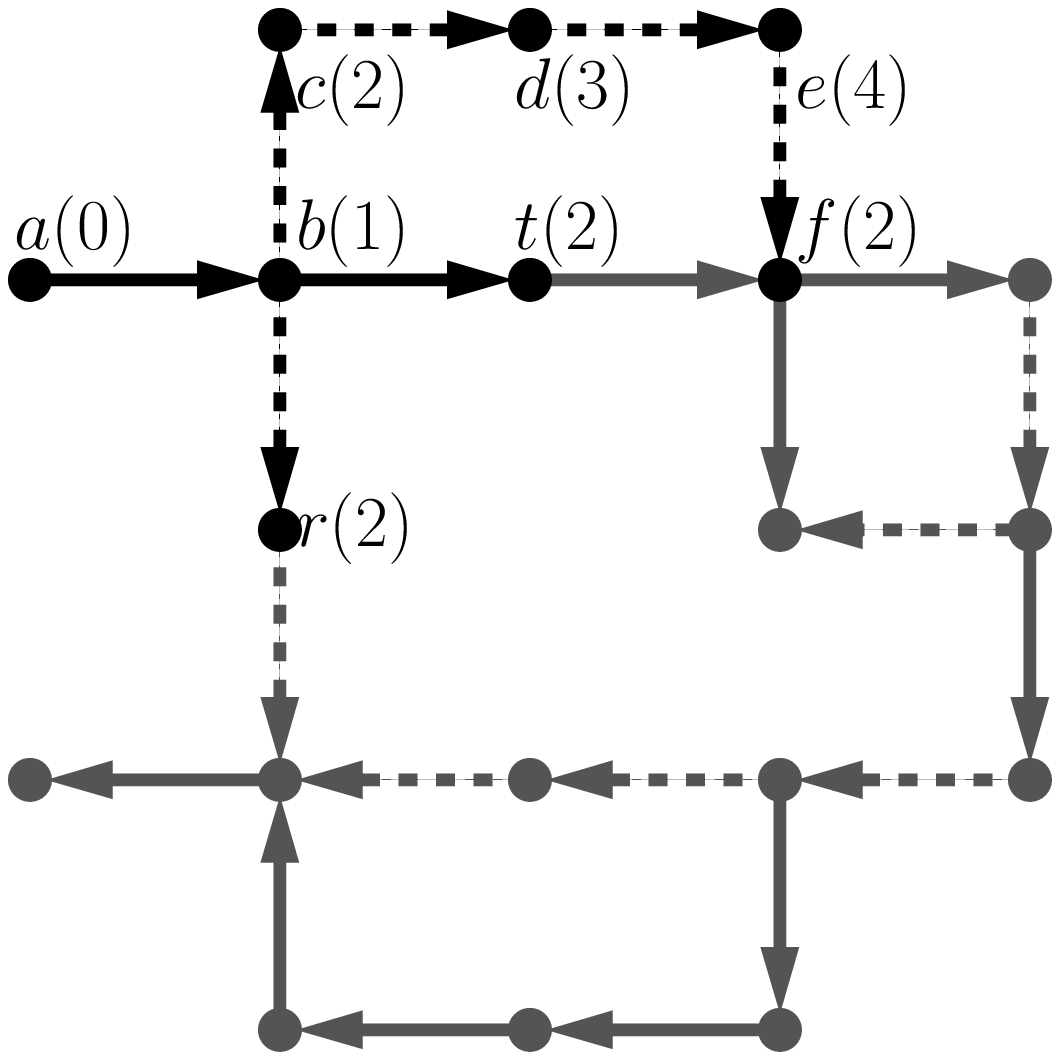, scale=0.33}
  \vskip 20pt
  \epsfig{file=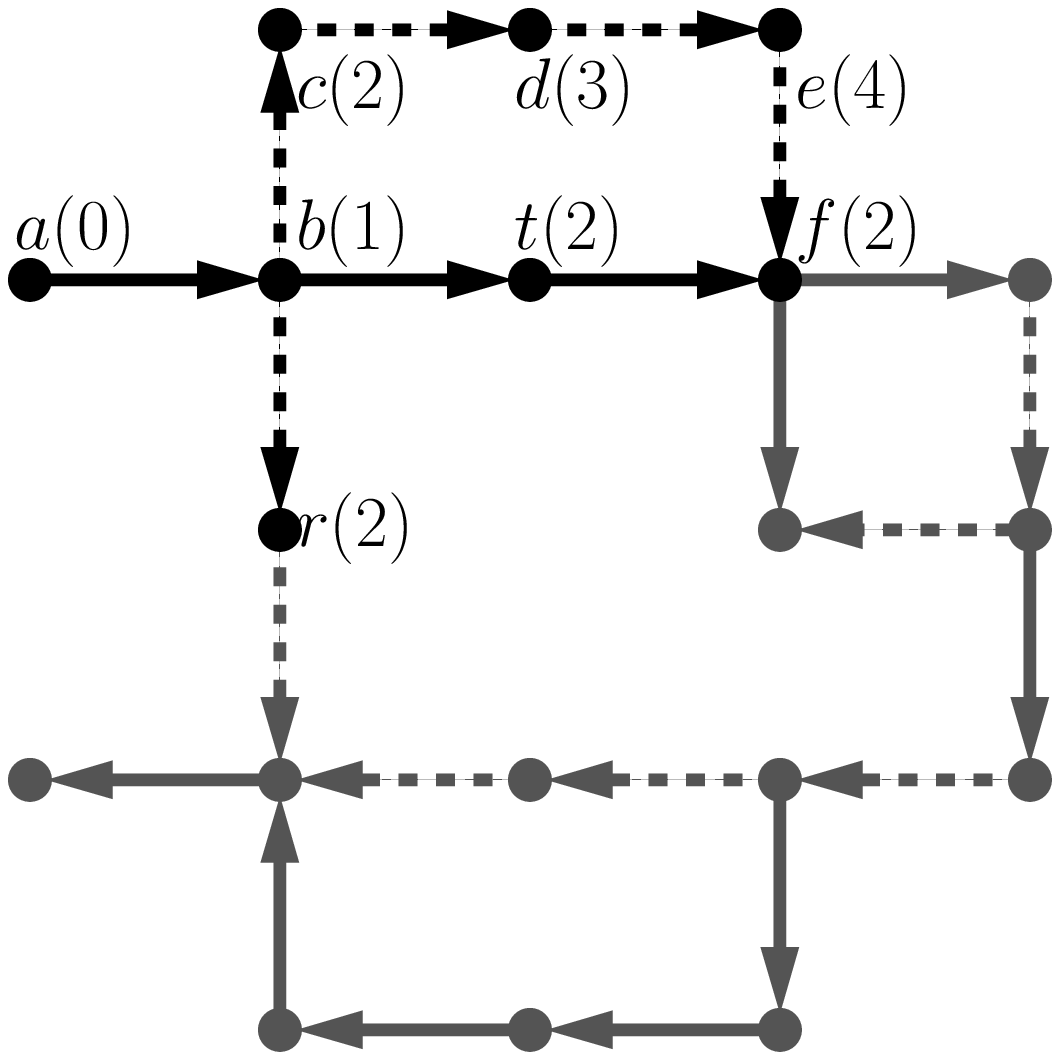, scale=0.33}
  ~
  \epsfig{file=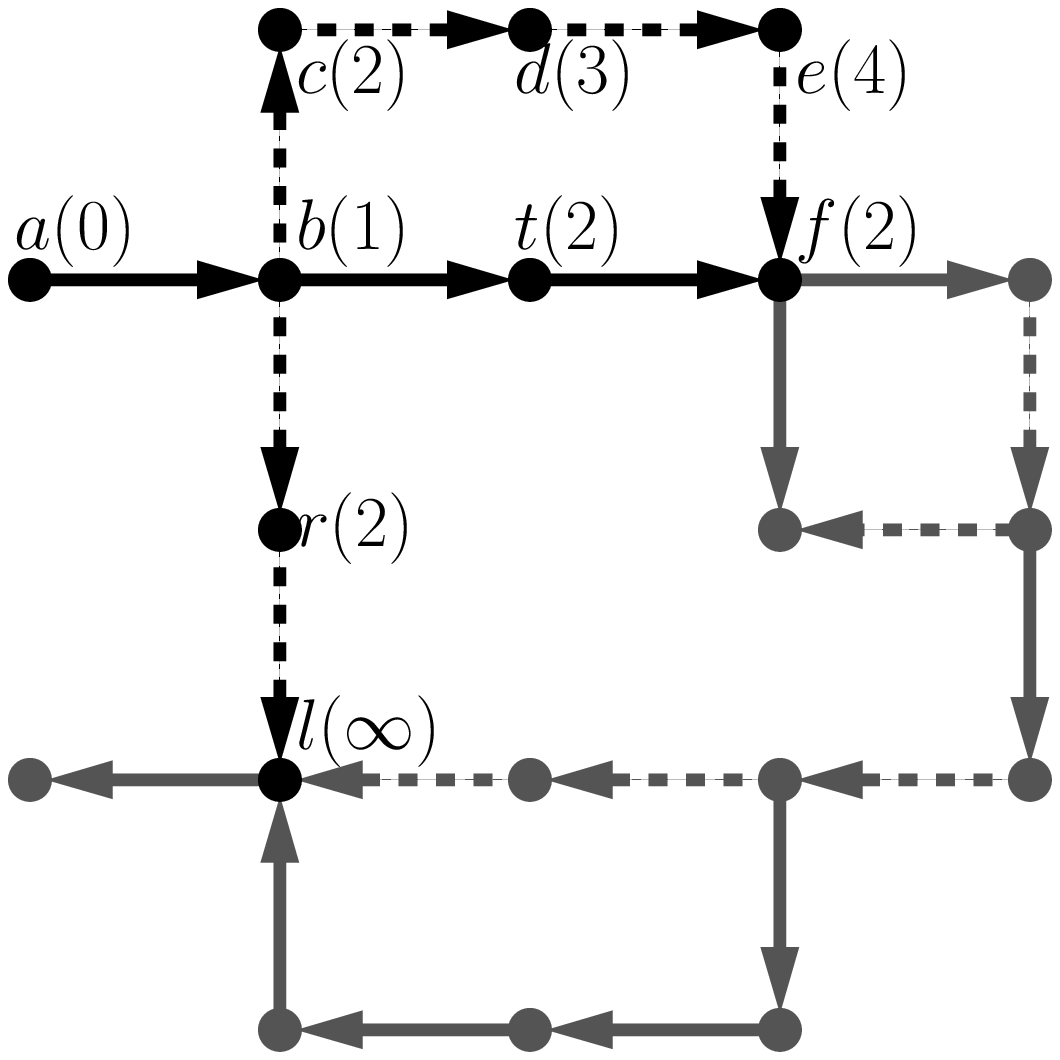, scale=0.33}
  ~
  \epsfig{file=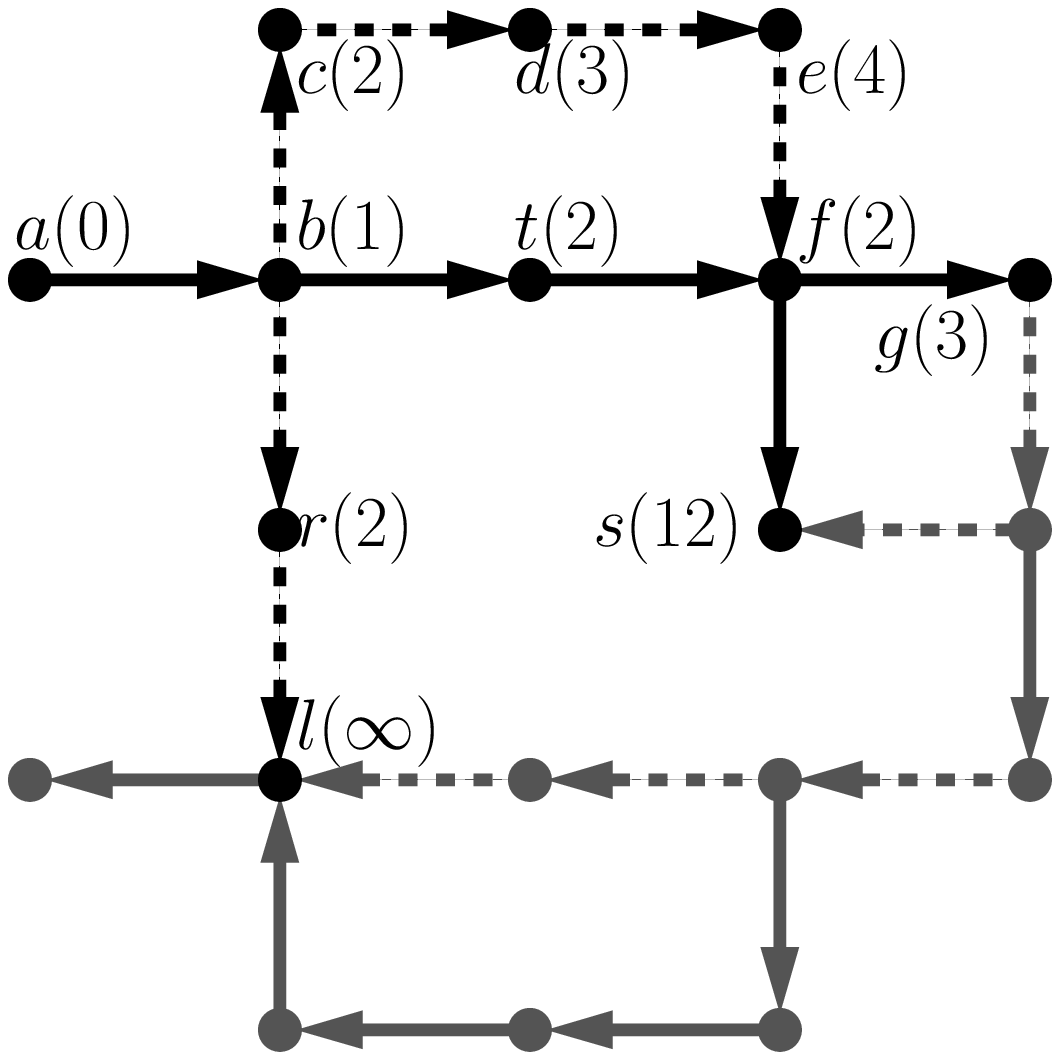, scale=0.33}
  \vskip 20pt
  \epsfig{file=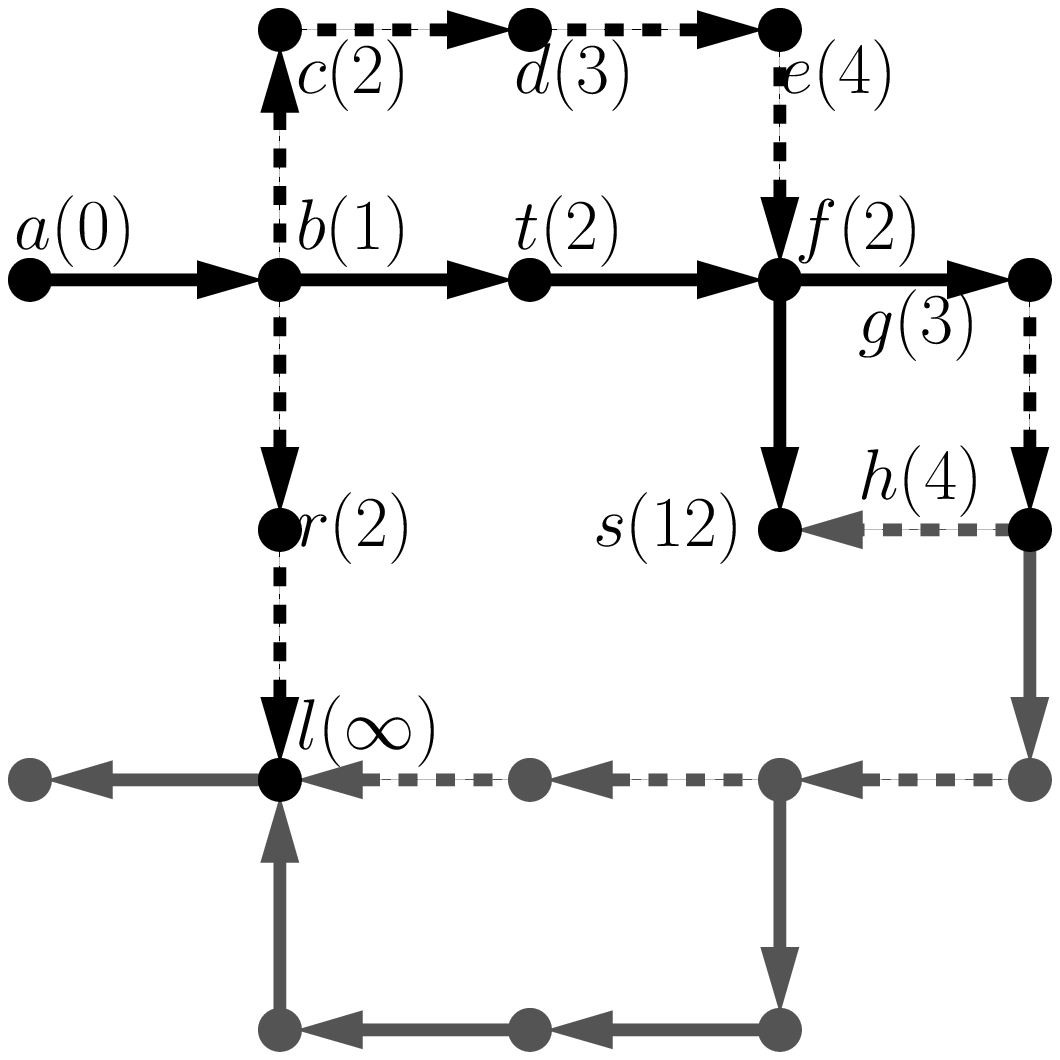, scale=0.33}
  ~
  \epsfig{file=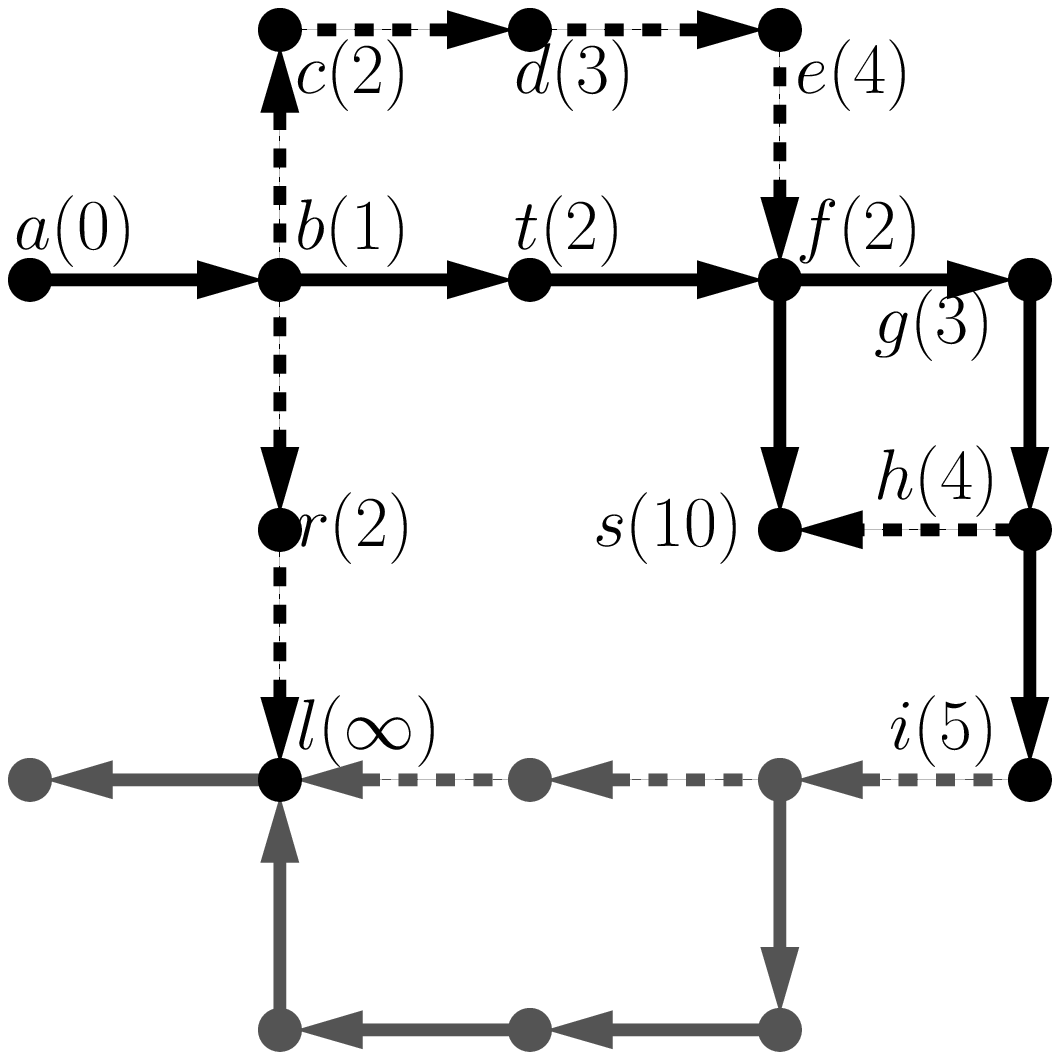, scale=0.33}  
  ~
  \epsfig{file=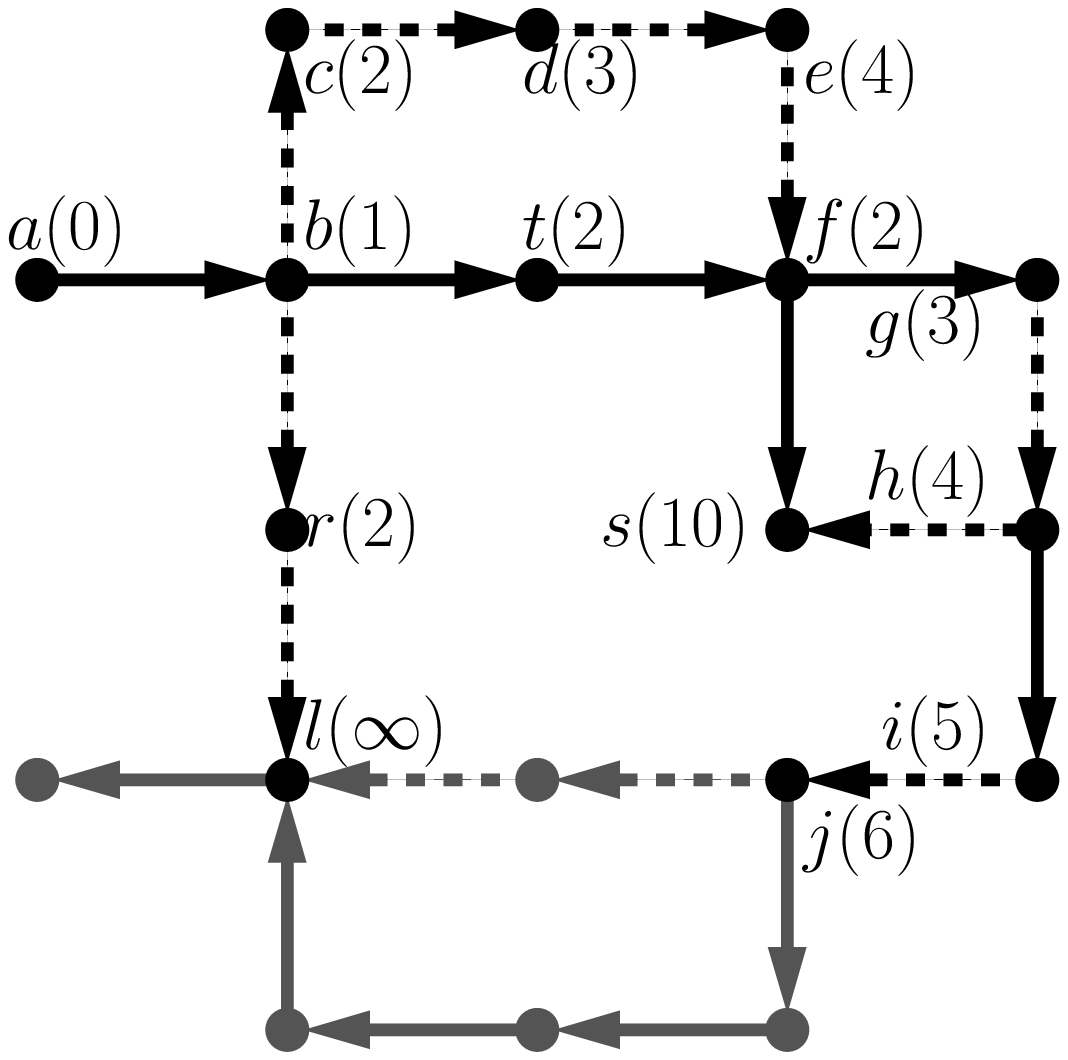, scale=0.33}
  \vskip 20pt
  \epsfig{file=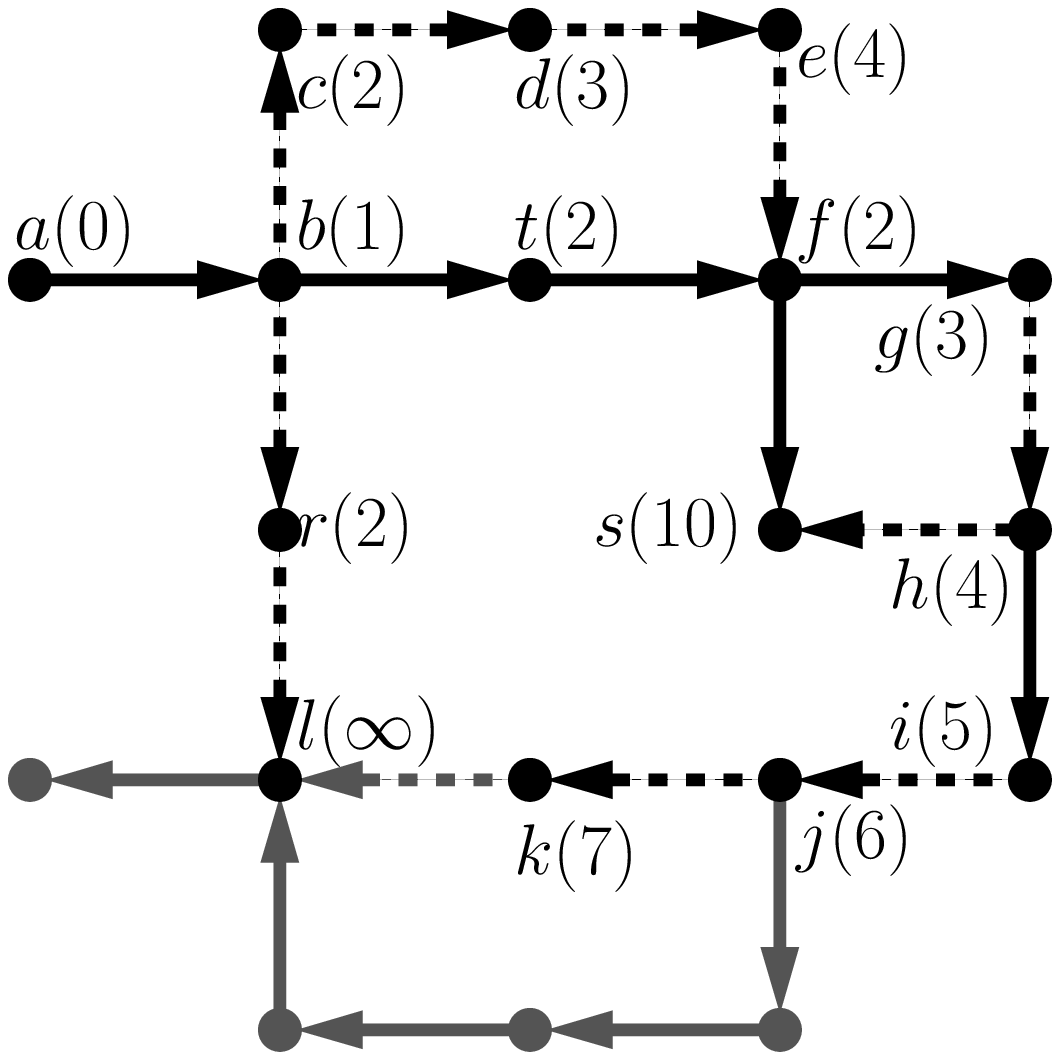, scale=0.33}
  ~
  \epsfig{file=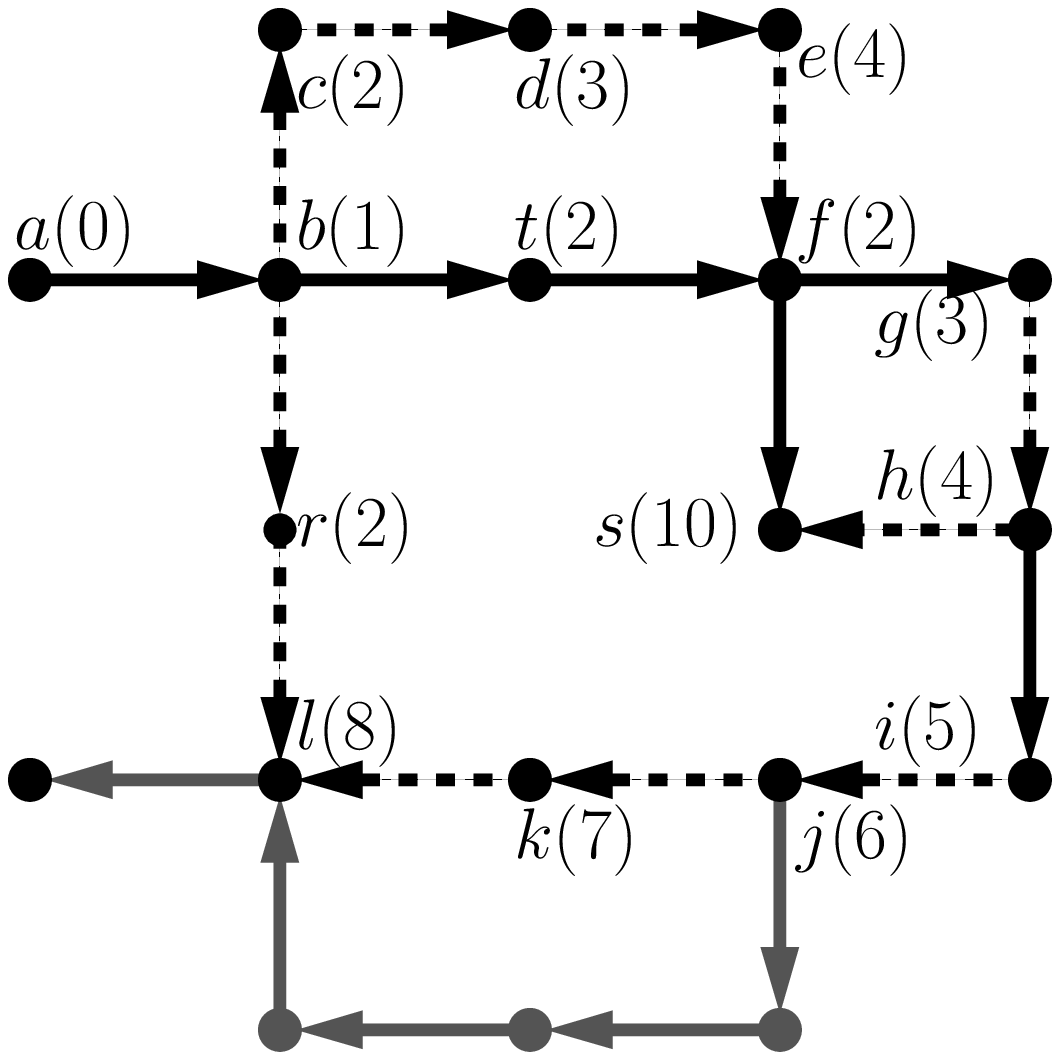, scale=0.33}  
  ~
  \epsfig{file=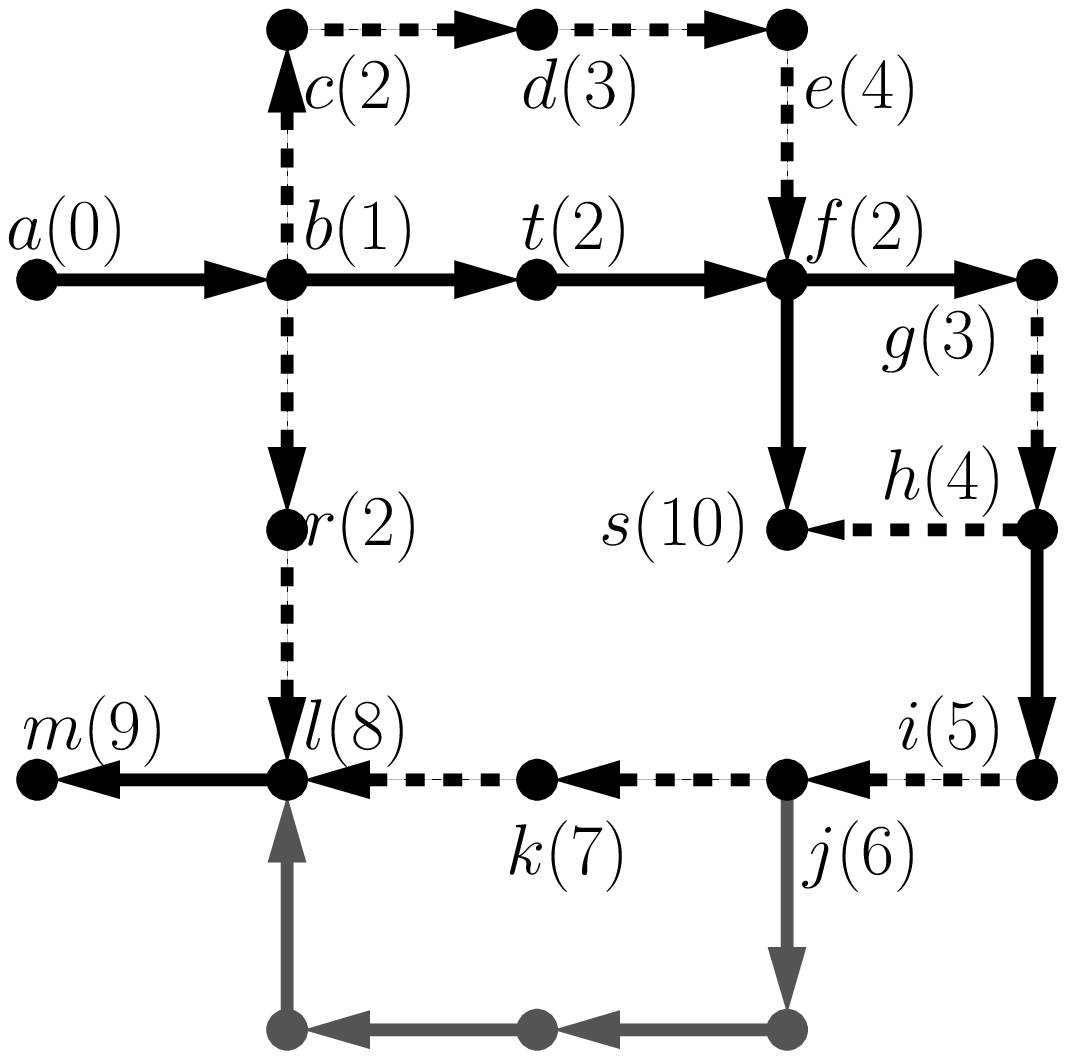, scale=0.33}
  \caption{A computation of an optimal walk w.r.t. the penalized weight from 
    $a$ to $m$ in $G$. Numbers represent the distance from the start $a$. 
    Black vertices are reached or scanned and black edges were relaxed.
    Dotted edges represent maneuver edges. 
    Steps 6 and 7 are depicted in the same figure (they are equal),
    analogously for steps 8 and 9.}
  \label{fig:example}
\end{figure}

\section{Conclusion}

We have introduced a novel generic model of maneuvers that is able to capture
almost arbitrarily complex route restrictions, traffic regulations and even
some dynamic aspects of the route planning problem. It can model anything 
from single vertices to long self-intersecting walks as restricted, negative, 
positive or prohibited maneuvers. We have shown how to incorporate this model 
into Dijkstra's algorithm so that no adjustment of the underlying road network 
graph is needed. The running time of the proposed Algorithm~\ref{alg:m-dijkstra} 
is only marginally larger than that of ordinary Dijkstra's algorithm 
(Theorem~\ref{thm:MDijkstra}) in practical networks. 

Our algorithm can be relatively straightforwardly extended to a bidirectional 
algorithm by running it simultaneously from the start vertex in the original 
network and from the target vertex in the reversed network. A termination
condition must reflect the fact that chained contexts of vertex-context pairs 
scanned in both directions might contain maneuvers as subwalks. 
Furthermore, since the A* algorithm is just an ordinary Dijkstra's algorithm 
with edge weights adjusted by a potential function, our extension remains
correct for A* if the road network is proper (Definition~\ref{def:road_network}, 
namely iii.) even with respect to this potential function.

Finally, we would like to highlight that, under reasonable assumptions, our 
model can be incorporated into many established route planning approaches.

\bibliographystyle{plain}
\bibliography{references}

\end{document}